\documentclass[submission,copyright,creativecommons]{eptcs}
\providecommand{\event}{MFPS 2021} 
\usepackage{breakurl}             
\usepackage{underscore}           

\usepackage{amsmath,amssymb,amsthm,prftree}

\newtheorem{theorem}{Theorem}[section]
\newtheorem{lemma}[theorem]{Lemma}
\newtheorem{corollary}[theorem]{Corollary}
\newtheorem{application}[theorem]{Application}
\theoremstyle{remark}
\newtheorem{remark}[theorem]{Remark}
\newtheorem{example}[theorem]{Example}

\DeclareMathOperator{\CZF}{\mathbf{CZF}}
\DeclareMathOperator{\Fin}{Fin}
\DeclareMathOperator{\Pow}{Pow}
\DeclareMathOperator{\id}{id}
\DeclareMathOperator{\PC}{PC}
\DeclareMathOperator{\EFQ}{EFQ}
\DeclareMathOperator{\RAA}{RAA}
\DeclareMathOperator{\DNS}{DNS}

\title{A General Glivenko--G{\"o}del Theorem for Nuclei}
\author{Giulio Fellin
	\institute{Universit\`a di Verona,\\ Strada le Grazie 15, 37134 Verona, Italy}
	\institute{Universit\`a di Trento, Italy}
	\institute{University of Helsinki, Finland}
	\email{giulio.fellin@univr.it}
	\and
	Peter Schuster
	\institute{Universit\`a di Verona,\\ Strada le Grazie 15, 37134 Verona, Italy}
	\email{peter.schuster@univr.it}
}
\def\titlerunning{A General Glivenko--G{\"o}del Theorem for Nuclei}
\def\authorrunning{G. Fellin \& P. Schuster}
\begin{document}
	\maketitle
	
	\begin{abstract}
		Glivenko's theorem says that, in propositional logic, classical provability of a formula entails intuitionistic provability of double negation of that formula. 
		We generalise Glivenko's theorem from double negation to an arbitrary nucleus, from provability in a calculus to an inductively generated 
		abstract consequence relation, and from propositional logic to any set of objects whatsoever. The resulting conservation theorem comes with 
		precise criteria for its validity,  which allow us to instantly include G{\"o}del's counterpart for first-order predicate logic of Glivenko's theorem. 
		The open nucleus gives us a form of the deduction theorem for positive logic, and 
		the closed nucleus prompts a variant of the reduction from intuitionistic to minimal logic going back to Johansson. 
	\end{abstract}
	
	\section{Introduction}
	Double negation over intuitionistic logic is a typical instance of a nucleus \cite{aczel:aspects,hayk:more,joh:sto,negri:cont,rosen,simm:frame,simm:cur,van:kuroda}.
	Glivenko's theorem says that, in propositional logic, classical provability of a formula entails intuitionistic provability of the double negation of that formula \cite{gliv29}.
	This stood right at the beginning of the success story of negative translations, which have been put into the context of nuclei \cite{van:kuroda}
	or monads \cite{esca:peirce}.
	As compared to recent literature on Glivenko's theorem 
	\cite{gue:pos,esp:gli,neg:gli,negri:barr,per:con,lit:neg,ono:gli,ish:emb,fel:jac:proc,10.2307/27588518},\footnote{This list of references is by no means meant exhaustive.}
	the purpose of the present paper is to generalise Glivenko's theorem from double negation to an arbitrary nucleus, 
	from provability in a calculus to an abstract consequence relation, and from propositional logic to any set of objects whatsoever. 
	
	To this end we move to a nucleus $j$ over a Hertz--Tarski consequence relation 
	in the form of a (single-conclusion) entailment relation $\rhd$ à la Scott \cite{sco:com,cedcoq:entaildistr}. Assuming that $\rhd$ 
	is inductively generated by axioms and rules, we propose two natural extensions 
	(Section \ref{induced}): $\rhd_j$ generalises the provability of double negation, 
	and $\rhd^j$ is inductively defined by adding the generalisation of double negation elimination to the inductive definition of $\rhd$.
	By their very definitions, $\rhd^j$ satisfies all axioms and rules of $\rhd$, and $\rhd_j$ satisfies all axioms of $\rhd$.
	But when does $\rhd_j$ also satisfy all rules of $\rhd$? Our main result, Theorem \ref{mcextendrules}, 
	says that $\rhd^j$ extends $\rhd_j$, and that the two relations coincide precisely when $\rhd_j$ 
	is closed under the non-axiom rules that are used to inductively generate $\rhd$, 
	which of course is the case whenever there are no such non-axiom rules (Corollary \ref{mcextend}). 
	
	In logic this gives us a multi-purpose conservation criterion (Theorem \ref{mainlogic}),  
	by which propositional and predicate logic can be handled in parallel. 
	The prime instance of course is Glivenko's theorem (Application \ref{Glivenko}(i)) as a syntactical conservation theorem (see also \cite{fel:jac:proc,fel:jac}): 
	\[\Gamma\vdash_c\varphi\iff\Gamma\vdash_i\neg\neg\varphi\]
	where ${\vdash_c}$ and ${\vdash_i}$ denote classical and intuitionistic propositional logic. 
	Simultaneously we re-obtain G\"odel's theorem (Application \ref{Glivenko}(ii)) which states that 
	\[\Gamma\vdash_c^Q\varphi\iff\Gamma\rhd_*^Q\neg\neg\varphi\]
	where ${\vdash_c^Q}$ denotes classical predicate logic, and $\rhd_*^Q$ is any extension (by additional axioms) of intuitionistic predicate logic that satisfies 
	the {double negation shift}: 
	\[\forall x\neg\neg\varphi\rhd\neg\neg\forall x\varphi\]
	While the double negation nucleus $j\varphi\equiv \lnot\lnot\varphi$ is an instance of the continuation monad,  
	it is tantamount to the same case $j\varphi\equiv \neg\varphi\to\varphi$ of the Peirce monad \cite{esca:peirce}. 
	What does our main result mean for other nuclei in logic? The Dragalin--Friedman nucleus  
	$j\varphi\equiv \varphi\vee \bot$, a case of the closed nucleus, yields a variant of the reduction 
	from intuitionistic to minimal logic going back 
	to Johansson (Application \ref{DF}). 
	Last but not least, the open nucleus $j\varphi\equiv A\to \varphi$ prompts a form of the deduction theorem 
	for positive logic (Application \ref{open}).

	\subsubsection*{Preliminaries}
	We intend to proceed in a constructive and predicative way, keeping the concepts elementary and the proofs direct. 
	If a formal system is desired, our work can be placed in a suitable fragment of Aczel's	
	\emph{Constructive Zermelo--Fraenkel Set Theory} ($\CZF$) \cite{a78,a82,a86,aczel:notes,aczel:cstdraft} 
	based on intuitionistic first-order predicate logic. 
	
	{By a \emph{finite set} we understand a set that can be written as $\{a_1, \dots, a_n\}$ for some $n \geq 0$. 
		Given any set $S$, let $\Pow(S)$ (respectively, $\Fin(S)$) consist of the (finite) subsets of $S$.
		We refer to \cite{rin:edde:full} for further provisos to carry over to the present note.}\footnote{
		For example, we deviate from the terminology prevalent in constructive mathematics and set theory 
		\cite{aczel:notes,aczel:cstdraft,Bishop67,Bishop-Bridges85,lombardiquitte:constructive,min:bib}: 
		to reserve the term `finite' to sets which are in \emph{bijection} with $\{1,\ldots,n\}$
		for a necessarily unique $n\ge 0$. Those exactly are the sets which are finite in our sense and are \emph{discrete} too, i.e.~have 
		decidable equality \cite{min:bib}. 
	} 
	
	\section{Entailment relations}
	
	Entailment relations 
	are at the heart of this note.
	We briefly recall the basic notions,
	closely following \cite{rin:edde,rin:edde:full}.
	
	Let $S$ be a set and $\rhd\subseteq\Pow(S)\times S$.
	Once abstracted from the context of logical formulae, all but one of Tarski's axioms of \emph{consequence} \cite{tarski:fundamental}
	\footnote{Tarski has further required that $S$ be countable.} 
	can be put as
	\begin{align*}
		\prftree
		{U\ni a}{U\rhd a} &&
		\prftree
		{\forall b\in U(V\rhd b)}{}{U\rhd a}{V\rhd a}  &&
		\prftree
		{U\rhd a}{\exists U_0\in\Fin(U)(U_0\rhd a)}
	\end{align*}
	where $U,V\subseteq S$ and $a\in S$. These axioms also characterise a {finitary covering} or {Stone covering}
	in formal topology \cite{sam:ifs};\footnote{This is from where we have taken the symbol $\rhd$,
		used also \cite{wang:lwd,cintula:pbc} 
		to denote a `{consecution}' \cite{restall:sub}.}
	see further \cite{Ciraulo:finitary,cir:conv,neg:sto,negri:cont,sam:som,sam:bas}. 
	The notion of consequence has presumably been described first by Hertz \cite{hertz:1,hertz:2,hertz:3}; see also \cite{bez:ax,leg:her}.
	
	Tarski has rather characterised the set of consequences of a set of propositions, which corresponds to
	the \emph{algebraic closure operator} $U\mapsto U^{\rhd}$ on $\Pow(S)$ of a relation $\rhd$ as above where 
	\[
	U^{\rhd}\equiv\{a\in S:U\rhd a\}\,.
	\]
	Rather than with Tarski's notion, we henceforth work with its (tantamount) restriction to finite subsets, i.e.~a 
	\emph{(single-conclusion) entailment relation}.
	\footnote{In the present paper there is no need for abstract \emph{multi-conclusion} consequence or entailment \`a la Scott \cite{sco:eng,sco:com,sco:bac}, 
		Lorenzen's contributions to which are currently under scrutiny \cite{cln:groups,coq:lorenzen}. 
		The relevance of multi-conclusion entailment to constructive algebra, point-free topology, etc.~has been pointed out in \cite{cedcoq:entaildistr}, 
		and has widely been used, e.g.~in \cite{coquand:dlhb,coq:stone,coq:valspace,coquand:lc,coq:hidden-krull,coq:prague,coq:seq,
			neg:ord,rin:phd,wes:ogc,rinwes:sik,sch:hah,LICS2020,wes:path,rin:edde,rin:edde:full,neg:ord,lombardiquitte:constructive}. }  
	This is a relation $\rhd\subseteq\Fin(S)\times S$
	such that 
	\begin{align*}
		\prftree[r]{(R)}{U\ni a}{U\rhd a} &&
		\prftree[r]{(T)}{V\rhd b}{}{V',b\rhd a}{V,V'\rhd a}  &&
		\prftree[r]{(M)}{U\rhd a}{U,U'\rhd a}
	\end{align*}
	for all finite $U,U',V,V'\subseteq S$ and $a,b\in S$, where as usual $U,V\equiv U\cup V$ {and} $V,b\equiv V\cup\{b\}$. 
	Our focus thus is on \emph{finite} subsets of $S$, for which we reserve the letters $U, V, W, \ldots$;
	we sometimes write $a_1,\ldots,a_n$ in place of $\{a_1,\ldots,a_n\}$ even if $n=0$. 
	
	\begin{remark}\label{rmkD}
		The rule (R) is equivalent, by (M), to the axiom $a\rhd a$.
	\end{remark}
	
	Redefining
	\begin{align*}
		T^{\rhd}\equiv\{a\in S:\exists U\in\Fin(T)(U\rhd a)\}
	\end{align*}
	for \emph{arbitrary} subsets $T$ of $S$ gives back an algebraic closure operator on $\Pow(S)$. 
	By writing $T \rhd a$ in place of $a \in T^{\rhd}$, 
	the entailment relations thus correspond exactly to the relations satisfying Tarski's axioms above.
	
	Given an entailment relation $\rhd$, by setting $a \leq b \equiv a \rhd b$ we get a preorder on $S$; 
	whence the conjunction $a \approx b$ of $a \leq b$ and $b \leq a$ is an equivalence relation. 
	
	Quite often an entailment relation is inductively generated from axioms by closing up with respect to 
	the three rules above \cite{rin:cuts}. 
	Some leeway is required in the present paper by allowing for generating rules other than (R), (M), and (T). 
	If, however, these three rules are the only rules employed for inductively generating an entailment relation, 
	we stress this by saying that this is \emph{generated only by axioms}. Given an inductively generated entailment relation 
	$\rhd$ and a set of axioms and rules $P$, then we call 
	\emph{$\rhd$ plus $P$} the entailment relation inductively generated by all axioms and 
	rules that either are used for generating $\rhd$ or belong to $P$. 
	
	A main feature of inductive generation is that 
	if $\rhd$ is an entailment relation generated inductively by certain axioms and rules, 
	then ${\rhd}\subseteq{\rhd'}$ for every entailment relation ${\rhd'}$ satisfying those axioms and rules. 
	By an \emph{extension} ${\rhd'}$ of an entailment relation ${\rhd}$ we mean in general 
	an entailment relation ${\rhd'}$ such that 
	${\rhd}\subseteq{\rhd'}$. We say that an extension ${\rhd'}$ of ${\rhd}$ is 
	\emph{conservative} if 
	also ${\rhd}\supseteq{\rhd'}$ and thus ${\rhd}={\rhd'}$ altogether \cite{rin:edde,rin:edde:full,fel:jac:proc,fel:jac}.
	
	\section{Nuclei over entailment relations}
	Throughout this section, fix a set $S$ endowed with an entailment relation $\rhd$.
	We say that a function $j\colon S\to S$ is a \emph{nucleus (over $\rhd$)} if for all $a,b\in S$ and $U\in\Fin(S)$ the following hold:
	\begin{align*}
		\prftree[r]{L$j$}{U,a\rhd jb}{U,ja\rhd jb}
		&&
		\prftree[r]{R$j$}{U\rhd b}{U\rhd jb}
	\end{align*}
	Unlike L$j$, by (R) and (T) the rule R$j$ can be expressed by an axiom, viz.
	\begin{align}
		b\rhd jb\,\label{nucleus1}
	\end{align}
	
	\begin{remark}
		The above notion of a nucleus includes as a special case the notion of a nucleus on a locale 
		\cite{aczel:aspects,joh:sto,negri:cont,rosen,simm:frame,simm:cur}, which is well-known as a point-free way to put subspaces. In fact, 
		if $S$ is a locale with partial order $\le$, then 
		\[U\rhd a\iff \bigwedge U\le a \]
		defines an entailment relation \cite{coq:seq} such that any given map $j:S\to S$ is a nucleus on $\rhd$ 
		precisely when $j$ is a nucleus on the locale $S$. The latter means 
		that $j$ satisfies 
		\begin{align}\label{meet}
			ja\wedge jb\le j(a\wedge b)
		\end{align}
		on top of the conditions for $j$ being a closure operator on $S$, which can be put as $a\le ja$ and 
		\begin{align}\label{clos}
			a\le jb\implies ja\le jb\,.
		\end{align}
		In the presence of $a\le ja$, which is nothing but (\ref{nucleus1}), the conjunction of (\ref{meet}) and (\ref{clos}) is equivalent to 
		\begin{align*}
			c\wedge a\le jb\implies c\wedge ja\le jb\,,
		\end{align*}
		which in turn subsumes L$j$. So the two notions of a nucleus coincide. 
	\end{remark}

	\begin{example}\label{exnuclei}\
		\begin{enumerate}
			\item Every  entailment relation $\rhd$ has the trivial nucleus $j\equiv\id$.
			\item Consider an algebraic structure $\mathbf{S}$ with a unary self-inverse function $j$
			(e.g. take a group as $\mathbf{S}$ and the inverse as $j$).
			The entailment relation $\rhd$ of $\mathbf{S}$-substructures is inductively defined by
			\begin{align}\label{substructure}
				a_1,...,a_n\rhd f(a_1,...,a_n)
			\end{align}
			for every $n$-ary function $f$ in the language of $\mathbf{S}$, including $j$.
			We want to show that $j$ is a nucleus on $\rhd$.
			Axiom (\ref{nucleus1}) is just (\ref{substructure}) for $f\equiv j$, therefore rule R$j$ holds.
			In particular, $j^2=\id$ implies $j(a)\rhd a$, which, together with (T), gives rule L$j$.
			In conclusion, $j$ is a nucleus on $\rhd$.
			\item 
			Double negation $\neg\neg$ is a nucleus over intuitionistic logic $\vdash_i$ as an entailment relation
			(see Subsection \ref{SSGlivenko} for further details and Subsections \ref{DragalinFriedman}--\ref{deduction} for more nuclei in logic).
		\end{enumerate}
	\end{example}
	
	\subsection*{Entailment relations induced by a nucleus, and conservation}\label{induced}
	Consider a nucleus $j$ over an entailment relation $\rhd$. We define
	\begin{itemize}
		\item[---] the \emph{weak $j$-extension} (or \emph{Kleisli extension}) of $\rhd$ as the relation ${\rhd_j}\subseteq\Fin(S)\times S$ defined by
		\begin{align*}
			U\rhd_ja\iff U\rhd ja
		\end{align*}
		\item[---] the \emph{strong $j$-extension} as the entailment relation ${\rhd^j}\subseteq\Fin(S)\times S$ inductively generated 
		by the axioms and rules of $\rhd$ plus the \emph{stability axiom} for $j$:
		\begin{align}
			ja\rhd^ja\label{extaxiom}
		\end{align}
		In the terminology coined before, ${\rhd^j}$ is nothing but $\rhd$ plus the stability axiom for $j$.  
	\end{itemize}
	\begin{remark}\label{rmkStab}
		By (R) in the form of $a\rhd a$ (Remark \ref{rmkD}), stability holds for $\rhd_j$ too, that is, $ja\rhd_ja$. 
	\end{remark}
	Under appropriate circumstances Remark \ref{rmkStab} will help to obtain ${\rhd^j}\subseteq{\rhd_j}$; see Theorem \ref{main} and Corollary \ref{mcextend}. 
	\begin{lemma}\label{rhdj}
		Let $S$ be a set with an entailment relation $\rhd$ and let $j$ be a nucleus on $\rhd$. 
		\begin{enumerate}
			\item 	$\rhd^j$ is an entailment relation that extends $\rhd$.
			\item	$\rhd_j$ is an entailment relation that extends $\rhd$.
		\end{enumerate}
	\end{lemma}
	\begin{proof}
		(i) holds by the very definition of $\rhd^j$.
		As for (ii): By (\ref{nucleus1}) and Remark \ref{rmkD}, rule (R) is bestowed from $\rhd$ to $\rhd_j$.
		Rule (M) is inherited from $\rhd$, and so is rule (T) in view of L$j$:
		\[
		\prftree[r]{(T)}{U\rhd ja}{
			\prftree[r]{L$j$}{V,a\rhd jb}{V,ja\rhd jb}
		}{U,V\rhd jb}
		\]
		Finally, also ${\rhd}\subseteq{\rhd_j}$ is a consequence of (\ref{nucleus1}). 
	\end{proof}
	
	\begin{remark}
		The nucleus $j$ on $\rhd$ is a nucleus also on $\rhd_j$ and $\rhd^j$. In fact, by Lemma \ref{rhdj} both extensions inherit axiom (\ref{nucleus1}) 
		from $\rhd$, and actually satisfy the following strengthening of L$j$: 
		\begin{align*}
			\prftree[r]{L$j^+$}{U,a\rhd b}{U,ja\rhd b}\,.
		\end{align*}
		While L$j^+$ for $\rhd_j$ is just $Lj$ for $\rhd$, stability $ja\rhd a$ is tantamount to L$j^+$ for any entailment relation $\rhd$ whatsoever. 
	\end{remark}
	
	To understand better whether and when $\rhd_j$ coincides with $\rhd^j$, we first consider a concrete example.
	\begin{example}\label{example}
		Consider deduction in minimal logic $\vdash_m$ with the nucleus $j\varphi\equiv\varphi\vee\bot$ (see Subsection \ref{DragalinFriedman} below for details). 
		Propositional minimal logic $\vdash_m$ is inductively generated by certain axioms plus the rule
		\[ \prftree[r]{R$\to$}{\Gamma,\varphi\vdash_m\psi}{\Gamma\vdash_m\varphi\to\psi} \]
		which cannot be expressed as an axiom. By its very definition, $\vdash_m^j$ too satisfies R$\to$. Does also $\vdash_{mj}$ satisfy this rule? If this were the case, 
		then by definition of $\vdash_{mj}$ we would have 
		\[ \prftree{\Gamma,\varphi\vdash_m\psi\vee\bot}{\Gamma\vdash_m(\varphi\to\psi)\vee\bot} \]
		As $\bot\vdash_m\psi\vee\bot$, we would obtain $\vdash_m(\bot\to\psi)\vee\bot$. However, since minimal logic has the disjunction property and neither disjunct is provable in general, this cannot be the case. So $\rhd_j$ does not satisfy rule R$\to$.
	\end{example}
	
	The moral of Example \ref{example} is that $\rhd$ may already have non-axiom rules, such as R$\to$, which carry over to $\rhd^j$ by its very definition, and thus need to hold in $\rhd_j$ too for the former to be conservative over the latter. 
	To deal with this issue, we say that a rule $r$ that holds for $\rhd$ is \emph{compatible} with $j$ if $r$ also holds for $\rhd_j$.
	
	\begin{remark}\label{rmkaxiom}\
		\begin{enumerate}
			\item Rules (R), (M), (T) are compatible with every nucleus $j$, by Lemma \ref{rhdj}.
			\item Every composition $r$ of compatible rules is compatible. 
			In fact, the derivation that gives $r$ in $\rhd$ can be translated smoothly into $\rhd_j$, as all applied rules are compatible.
			
			This is very useful: if we want to check compatibility for all rules of an entailment relation $\rhd$, 
			it suffices to check compatibility for any set of rules that generate $\rhd$. 
			
			\item Every axiom $a_1,...,a_n\rhd b$ can be viewed as a rule with no premiss, and as such is compatible with 
			every nucleus $j$, simply by R$j$.
			Moreover, rules
			\[ \prftree{U,b\rhd c}{U,a_1,...,a_n\rhd c} \quad\quad \prftree{U\rhd a_1}{...}{U\rhd a_n}{U\rhd b} \]
			which are known respectively as \emph{left} and \emph{right rule} 
			\cite{rin:cuts,fel:msc}
			\footnote{A reader familiar with structural proof theory may be reminded of the notion of 
				left and right rules in sequent calculus \cite{spt,neg:pa}. Though they look similar, 
				the two concepts are not to be confused.} 
			are provably equivalent to the axiom $a_1,...,a_n\rhd b$ and therefore are compatible with $j$.
			\item If an entailment relation $\rhd$ is generated only by axioms, then every rule that holds for $\rhd$ is compatible with any nucleus $j$ over $\rhd$. 
		\end{enumerate}
	\end{remark}
	
	\begin{theorem}[Conservation for nuclei]\label{mcextendrules}\label{main}
		Let $S$ be a set with an entailment relation $\rhd$ inductively generated by axioms and rules, and let $j$ be a nucleus on $\rhd$.
		Then ${\rhd^j}$ extends ${\rhd_j}$, that is ${\rhd_j}\subseteq{\rhd^j}$. Moreover, the following are equivalent:
		\begin{enumerate}
			\item[\emph{(a)}] ${\rhd^j}$ is conservative over ${\rhd_j}$, that is, ${\rhd^j}\subseteq{\rhd_j}$;
			\item[\emph{(b)}] All non-axiom rules that generate $\rhd$ are compatible with $j$.
		\end{enumerate}
	\end{theorem}
	\begin{proof}
		First recall that, by its very definition, $\rhd^j$ is inductively generated by rules (R), (M), (T), stability (\ref{extaxiom}), and all rules that generate $\rhd$. In particular, ${\rhd}\subseteq{\rhd^j}$.
		
		Now suppose that $U\rhd_jb$, i.e.~$U\rhd jb$. Since ${\rhd}\subseteq{\rhd^j}$, also $U\rhd^jjb$.
		Then apply
		\[ \prftree[r]{(T)}{U\rhd^jjb}{}{jb\rhd^jb}{U\rhd^jb} \]
		to show ${\rhd_j}\subseteq{\rhd^j}$.
		
		\underline{(a)$\Rightarrow$(b)} (b) follows directly from (a) and the fact that $\rhd^j$ satisfies all rules that generate $\rhd$.
		
		\underline{(b)$\Rightarrow$(a)}
		Let us consider one by one the axioms and rules that generate $\rhd^j$:
		\begin{itemize}
			\item[---] $\rhd_j$ satisfies (R), (M), (T), since $\rhd_j$ is an entailment relation 
			by Lemma \ref{rhdj}.
			\item[---] $\rhd_j$ satisfies stability (\ref{extaxiom}) by Remark \ref{rmkStab}.
			\item[---] $\rhd_j$ satisfies all rules that generate $\rhd$ since they are either compatible with $j$ by hypothesis or axioms and thus compatible with $j$ by Remark \ref{rmkaxiom}.
		\end{itemize}
		As $\rhd^j$ is the smallest extension of $\rhd$ satisfying these axioms and rules, we get ${\rhd^j}\subseteq{\rhd_j}$.
	\end{proof}
	
	\begin{corollary}\label{mcextend}
		Let $S$ be a set with an entailment relation $\rhd$ inductively generated only by axioms, and let $j$ be a nucleus on $\rhd$.
		Then ${\rhd^j}={\rhd_j}$, that is, $\rhd^j$ is a conservative extension of $\rhd_j$.
	\end{corollary}
	
	Let $j$ be a nucleus over an entailment relation $\rhd$ inductively generated by axioms and rules, 
	and let $\rhd_*$ be an {extension} of $\rhd$. 
	We say that $\rhd_*$ is an \emph{intermediate $j$-extension} of $\rhd$ if $\rhd_*$ is $\rhd$ plus $*$ 
	where $*$ is a collection of axioms that are valid in $\rhd^j$. In particular, ${\rhd}\subseteq{\rhd_*}\subseteq{\rhd^j}$.
	\begin{remark}\label{rmkintermediate}
		Since $\rhd\subseteq\rhd_*$, we have ${\rhd^j}\subseteq{\rhd_*^j}$.
		On the other hand, as all axioms in $*$ already hold for $\rhd^j$, we also have ${\rhd_*^j}\subseteq{\rhd^j}$.
		Therefore ${\rhd_*^j}={\rhd^j}$.
	\end{remark}
	\begin{corollary}[Conservation for intermediate $j$-extensions]\label{intermediate}
		Let $S$ be a set with an entailment relation $\rhd$ inductively generated by axioms and rules, let $j$ be a nucleus on $\rhd$,
		and let $\rhd_*$ be an intermediate $j$-extension of $\rhd$.
		Then ${\rhd^j}$ extends ${\rhd_{*j}}$, that is ${\rhd_{*j}}\subseteq{\rhd^j}$. Moreover, the following are equivalent:
		\begin{enumerate}
			\item[\emph{(a)}] ${\rhd^j}$ is conservative over ${\rhd_{*j}}$, that is, ${\rhd^j}\subseteq{\rhd_{*j}}$;
			\item[\emph{(b)}] All non-axiom rules that generate $\rhd$ hold for $\rhd_{*j}$.
		\end{enumerate}
	\end{corollary}
	\begin{proof}
		Follows from Theorem \ref{main} for $\rhd_*$ by noticing that ${\rhd_*^j}={\rhd^j}$ (Remark \ref{rmkintermediate}) and 
		that all additional rules of $\rhd_*$ are axioms and thus already compatible with $j$ (Remark \ref{rmkaxiom}).
	\end{proof}
	
	The following characterisation will prove useful in several applications: 
	\begin{lemma}\label{proposition}
		Let $S$ be a set with an entailment relation $\rhd$, and let $j$ be a nucleus on $\rhd$.
		Let $r$ be a rule holding for $\rhd$. The following are equivalent:
		\begin{enumerate}
			\item[\emph{(a)}] Rule $r$ is compatible with $j$.
			\item[\emph{(b)}]
			For every instance
			\[
			\raisebox{-9pt}{$\prftree
				{U_1\rhd b_1}{...}{U_n\rhd b_n}{U\rhd b}$}
			\]
			of rule $r$, there is $\beta\in S$ such that $\beta\rhd jb$ and
			\begin{align}\label{beta}
				\raisebox{-9pt}{$\prftree
					{U_1\rhd jb_1}{...}{U_n\rhd jb_n}{U\rhd\beta}$}
			\end{align}
		\end{enumerate}
	\end{lemma}
	\begin{proof}
		\underline{(a)$\Rightarrow$(b)} If we take $\beta\equiv jb$, then (b) immediately follows by reflexivity and compatibility.
		
		\underline{(b)$\Rightarrow$(a)} Recall that $b\rhd jb$, and that from $U\rhd\beta$ and $\beta\rhd jb$ follows $U\rhd jb$ by (T).
	\end{proof}
	
	\section{Logic as entailment}
	
	\label{logicsection}
	Throughout this section, the overall assumption is that $S$ is a set of propositional or (first-order) predicate formulae 
	containing $\top,\bot$, and closed under the connectives $\vee,\wedge,\to,\neg$ for propositional logic 
	and also under the quantifiers $\forall,\exists$ for predicate logic. 
	Following \cite{ras:alg,blkov2019proof}, 
	by \emph{(propositional) positive logic} $\vdash_p$ we mean the positive fragment of propositional intuitionistic logic. 
	More precisely, we define $\vdash_p$ as the least entailment relation $\rhd$ that satisfies the \emph{deduction theorem}
	\begin{align*}
		\prftree[r]{R$\to$}{\Gamma,\varphi\rhd\psi}{\Gamma\rhd\varphi\to\psi}
	\end{align*}
	and the following axioms:
	\begin{align*}
		\varphi,\psi&\rhd\varphi\wedge\psi
		&\varphi\wedge\psi&\rhd\varphi
		&\varphi\wedge\psi&\rhd\psi
		\\
		\varphi&\rhd\varphi\vee\psi
		&\psi&\rhd\varphi\vee\psi
		&\varphi\vee\psi,\varphi\to\delta,\psi\to\delta&\rhd\delta
		\\
		\varphi,\varphi\to\psi&\rhd\psi
		\\
		&\rhd\top
	\end{align*}
	Of course, we understand this as an inductive definition.
	The above system for positive logic \cite{ras:alg} is equivalent to the $\mathbf{G3}$-style calculus in Table \ref{calculus} taken from 
	\cite{blkov2019proof}; they inductively generate the same entailment relation.
	\begin{table}
		\begin{align*}
			\hline\\
			&
			\prftree[r]{L$\wedge$}{\Gamma,\varphi,\psi\rhd\delta}{\Gamma,\varphi\wedge\psi\rhd\delta}
			&&
			\prftree[r]{R$\wedge$}{\Gamma\rhd\varphi}{\Gamma\rhd\psi}{\Gamma\rhd\varphi\wedge\psi}
			\\\vspace{10pt}
			&
			\prftree[r]{L$\vee$}{\Gamma,\varphi\rhd\delta}{}{\Gamma,\psi\rhd\delta}{\Gamma,\varphi\vee\psi\rhd\delta}
			&&
			\prftree[r]{R$\vee_1$}{\Gamma\rhd\varphi}{\Gamma\rhd\varphi\vee\psi}
			&
			\prftree[r]{R$\vee_2$}{\Gamma\rhd\psi}{\Gamma\rhd\varphi\vee\psi}
			\\\vspace{10pt}
			&
			\prftree[r]{L$\to$}{\Gamma\rhd\varphi}{}{\Gamma,\psi\rhd\delta}{\Gamma,\varphi\to\psi\rhd\delta}
			&&
			\prftree[r]{R$\to$}{\Gamma,\varphi\rhd\psi}{\Gamma\rhd\varphi\to\psi}
			\\\vspace{10pt}
			&
			\prftree[r]{R$\top$}{\phantom{x}}{\Gamma\rhd\top}
			\\\\\hline
		\end{align*}
		\caption{Sequent calculus-like rules for positive propositional logic \cite{ras:alg} following \cite{blkov2019proof}.}\label{calculus}
	\end{table}
	
	On top of $\vdash_p$ we consider the following additional axioms:
	\begin{align}
		\varphi\to\bot&\approx\neg\varphi\tag{$\PC$}
		\\
		\bot&\rhd\varphi\tag{$\EFQ$}
		\\
		\neg\neg\varphi&\rhd\varphi\tag{$\RAA$}
	\end{align}
	They are known as \emph{principium contradictionis}, \emph{ex falso quodlibet sequitur} and \emph{reductio ad absurdum}.
	The two directions of $\PC$ can also be expressed via the rules
	\begin{align*}
		\prftree[r]{L$\neg$}{\Gamma\rhd\varphi}{\Gamma,\bot\rhd\psi}{\Gamma,\neg\varphi\rhd\psi}
		&&
		\prftree[r]{R$\neg$}{\Gamma,\varphi\rhd\bot}{\Gamma\rhd\neg\varphi}
	\end{align*}
	In the presence of $\EFQ$, the rule L$\neg$ can be simplified as
	\begin{align*}
		\prftree[r]{L$\neg$}{\Gamma\rhd\varphi}{\Gamma,\neg\varphi\rhd\psi}
	\end{align*}
	Axiom $\EFQ$ is sometimes considered as a rule without premises:
	\[ \prftree[r]{L$\bot$}{\Gamma,\bot\rhd\varphi} \]
	
	We define:
	\begin{itemize}
		\item[---] \emph{minimal logic} $\vdash_m$ as $\vdash_p$ plus $\PC$,
		\item[---] \emph{intuitionistic logic} $\vdash_i$ as $\vdash_m$ plus $\EFQ$,
		\item[---] \emph{classical logic} $\vdash_c$ as $\vdash_i$ plus $\RAA$.
	\end{itemize}
	
	Let $\vdash_*$ be $\vdash_p$ plus additional axioms. In particular, $\vdash_*$ satisfies the deduction theorem R$\to$. 
	The (first-order) predicate version $\vdash_*^Q$ of $\vdash_*$, which we also refer to as $\vdash_*$ \emph{plus quantifiers},  
	is then obtained by adding quantifiers $\forall$ and $\exists$ to the language and the following rules to the inductive definition of $\vdash_*$:
	\begin{align*}
		\prftree[r]{L$\forall$}{\varphi[t/x],\Gamma,\forall x\varphi\rhd\delta}{\Gamma,\forall x\varphi\rhd\delta}
		&&\prftree[r]{R$\forall$}{\Gamma\rhd\varphi[y/x]}{\Gamma\rhd\forall x\varphi}
		\\
		\prftree[r]{L$\exists$}{\Gamma,\varphi[y/x]\rhd\delta}{\Gamma,\exists x\varphi\rhd\delta}
		&&\prftree[r]{R$\exists$}{\Gamma\rhd\varphi[t/x]}{\Gamma\rhd\exists x\varphi}
	\end{align*}
	with the condition that $y$ has to be fresh in L$\exists$ and R$\forall$. Rules L$\forall$ and R$\exists$ can be expressed as axioms: 
	\begin{align*}
		\forall x\varphi&\rhd\varphi[t/x]
		\\\varphi[t/x]&\rhd\exists x\varphi
	\end{align*}
	The definition of a nucleus $j$ given in \cite{van:kuroda} requires $j$ to be compatible with substitution, that is, 
	\[ j(\varphi[t/x])\equiv(j\varphi)[t/x] \]
	We prefer not to have this as a general assumption, but to make explicit whenever we need it.
	
	\section{Conservation for nuclei in logic}
	
	Among the usual logical rules, R$\to$, R$\forall$ and L$\exists$ are the only ones that cannot be expressed as axioms. Rule L$\exists$ is
	compatible with $j$ for every nucleus $j$ as it does not affect the right-hand side of the sequent.
	Therefore, when checking compatibility of rules with $j$, if we do not add other rules that cannot be expressed as axioms, then the only rules we have to check are R$\to$ and R$\forall$.
	\begin{lemma}\label{lemmacompatible}
		Let $\vdash_*$ be $\vdash_p$ plus additional axioms, and let $j$ be a nucleus on ${\vdash_*}$. Consider $\vdash_*$ as $\rhd$.
		\begin{enumerate}
			\item R$\to$ is compatible with $j$ if and only if \[\varphi\to j\psi\vdash_* j(\varphi\to\psi)\]
			\item If $j$ is compatible with substitution, then R$\forall$ is compatible with $j$ if and only if \[\forall xj\varphi\vdash_*^Q j\forall x\varphi\]
		\end{enumerate}
	\end{lemma}
	\begin{proof}
		We prove (i), the proof of (ii) is analogous. As for ``if'', by Lemma \ref{proposition}, R$\to$ is compatible with $j$ if and only if for 
		every instance
		\[
		\prftree
		{\Gamma,\varphi\vdash_*\psi}{\Gamma\vdash_*\varphi\to\psi}
		\]
		of R$\to$ there is $\beta\in S$ such that $\beta\vdash_* j(\varphi\to\psi)$ and
		\[
		\prftree
		{\Gamma,\varphi\vdash_* j\psi}{\Gamma\vdash_*\beta}
		\]
		By R$\to$, the latter condition is satisfied if we set $\beta\equiv\varphi\to j\psi$, 
		for which the former condition reads as
		\[\varphi\to j\psi\vdash_* j(\varphi\to\psi).\]
		
		As for ``only if'', compatibility directly 
		entails the desired criterion. In fact, as an 
		instance of \emph{modus ponens} we have 
		\[\varphi\to j\psi, \varphi \vdash_* j\psi,\]
		which by the very definition of $\vdash_j$ is nothing but 
		\[\varphi\to j\psi, \varphi \vdash_{*j} \psi.\]
		By compatibility, the deduction theorem carries over from $\vdash_{*}$ to $\vdash_{*j}$. Hence we get 
		\[\varphi\to j\psi \vdash_{*j} \varphi \to \psi,\]
		which again by the definition of $\vdash_j$ yields the desired criterion: 
		\[\varphi\to j\psi \vdash_{*} j(\varphi \to \psi). \]
	\end{proof}
	
	This gives us the following version of Corollary \ref{intermediate}:
	\begin{theorem}
		[Conservation for nuclei in logic]\label{mainlogic}
		Let $\vdash$ be $\vdash_p$ plus additional axioms, 
		let $j$ be a nucleus on ${\vdash}$,
		and 
		let ${\vdash_*}$ be $\vdash$ plus additional axioms such that ${\vdash_*}\subseteq{\vdash^j}$.
		\begin{enumerate}
			\item The following are equivalent in propositional logic:
			\begin{enumerate}
				\item $\Gamma\vdash^j\varphi\iff\Gamma\vdash_*j\varphi$ for all $\Gamma,\varphi$
				\item $\vdash_*$ satisfies the following axiom:
				\[\varphi\to j\psi\vdash_* j(\varphi\to\psi)\]
			\end{enumerate}
			\item 
			Let $\vdash^Q$, $\vdash_*^Q$, $\vdash^{Qj}$ be $\vdash$, $\vdash_*$, $\vdash^j$ plus quantifiers.
			If $j$ is compatible with substitution, then the following are equivalent in predicate logic:
			\begin{enumerate}
				\item $\Gamma\vdash^{Qj}\varphi\iff\Gamma\vdash_*^Qj\varphi$ for all $\Gamma,\varphi$
				\item $\vdash_*^Q$ satisfies the following axioms:
				\begin{align*}
					\varphi\to j\psi&\vdash_*^Q j(\varphi\to\psi)
					\\\forall xj\varphi&\vdash_*^Q j\forall x\varphi
				\end{align*}
			\end{enumerate}
		\end{enumerate}
	\end{theorem}
	
	\subsection{The Glivenko nucleus}\label{SSGlivenko}
	Take intuitionistic logic $\vdash_i$ as $\rhd$, 
	and define \[j\varphi\equiv\neg\neg\varphi\,.\]
	This $j$ is well-known to be a nucleus over $\vdash_i$ \cite{van:kuroda,rosen},
	which we call the \emph{Glivenko nucleus}. 
	As stability (\ref{extaxiom}) equals $\RAA$, the strong extension $\vdash_i^j$ of intuitionistic logic $\vdash_i$ is nothing but classical logic $\vdash_c$.
	
	Since $\varphi\to\neg\neg\psi\vdash_i\neg\neg(\varphi\to\psi)$ follows, e.g., 
	from \cite[Lemma 6.2.2]{dvd:las}, and the Glivenko nucleus is  compatible with substitution,
	we get the following instance of Theorem \ref{mainlogic} where $\vdash$ is $\vdash_i$:
	\begin{application}\label{Glivenko}\
		\begin{enumerate}
			\item (\textbf{Glivenko's Theorem}) $\Gamma\vdash_c\varphi\iff\Gamma\vdash_i\neg\neg\varphi$ for all $\Gamma,\varphi$ in propositional logic. 
			\item (\textbf{G\"odel's Theorem})
			Let $\vdash_{*}$ 
			be $\vdash_i$ plus additional axioms such that ${\vdash_*}\subseteq{\vdash_c}$, and let 
			$\vdash_i^Q$, $\vdash_*^Q$ and $\vdash_c^Q$ be $\vdash_i$, $\vdash_*$ and $\vdash_c$ plus quantifiers.
			The following are equivalent in 
			predicate logic:
			\begin{enumerate}
				\item $\Gamma\vdash_c^Q\varphi\iff\Gamma\vdash_*^Q\neg\neg\varphi$ for all $\Gamma,\varphi$; 
				\item $\forall x\neg\neg\varphi\vdash_*^Q\neg\neg\forall x\varphi$ for all $\varphi$. 
			\end{enumerate}
		\end{enumerate}
	\end{application}
	Condition (b) in Application \ref{Glivenko} is called \emph{Double Negation Shift} ($\DNS$) and is known to define a proper 
	intermediate logic $\vdash_{\DNS}^Q$, that is, $\vdash_i^Q\subsetneq\vdash_{\DNS}^Q\subsetneq\vdash_c^Q$ \cite{esp:gli}.
	
	Now let $j\varphi\equiv\neg\varphi\to\varphi$.
	This $j$ is a nucleus
	\cite{van:kuroda,rosen},
	which we call the \emph{Peirce nucleus}, as it is a special case of the Peirce monad \cite{esca:peirce}.
	Over intuitionistic logic, it is easy to show that the Glivenko nucleus is equivalent to the Peirce nucleus, i.e., 
	$\neg\neg\varphi\approx_i\neg\varphi\to\varphi$ 
	for every $\varphi$.
	
	\subsection{The Dragalin--Friedman nucleus}\label{DragalinFriedman}
	Take minimal logic $\vdash_m$ as $\rhd$, 
	and define \[j\varphi\equiv\varphi\vee\bot\,.\]
	This $j$ is a nucleus, in fact
	a \emph{closed nucleus} \cite{van:kuroda,rosen}. We refer to this $j$ as the \emph{Dragalin--Friedman nucleus}.
	As stability (\ref{extaxiom}) is  equivalent to $\EFQ$, the strong extension $\vdash_m^j$ of minimal logic $\vdash_m$ 
	is nothing but intuitionistic logic $\vdash_i$.
	
	Since the Dragalin--Friedman nucleus is  compatible with substitution, we get the following instance of Theorem \ref{mainlogic} 
	where $\vdash$ is $\vdash_m$: 
	\begin{application}\label{DF}
		Let $\vdash_{*}$ 
		be $\vdash_m$ plus additional axioms such that ${\vdash_*}\subseteq{\vdash_i}$.
		\begin{enumerate}
			\item The following are equivalent in propositional logic:
			\begin{enumerate}
				\item $\Gamma\vdash_i\varphi\iff\Gamma\vdash_*\varphi\vee\bot$ for all $\Gamma,\varphi$; 
				\item $\varphi\to(\psi\vee\bot)\vdash_*(\varphi\to\psi)\vee\bot$ for all $\varphi,\psi$. 
			\end{enumerate}
			\item Let $\vdash_m^Q$, $\vdash_*^Q$ and $\vdash_i^Q$ be $\vdash_m$, $\vdash_*$ and $\vdash_i$ plus quantifiers.
			The following are equivalent in predicate logic:
			\begin{enumerate}
				\item $\Gamma\vdash_i^{Q}\varphi\iff\Gamma\vdash_*^{Q}\varphi\vee\bot$ for all $\Gamma,\varphi$; 
				\item $\varphi\to(\psi\vee\bot)\vdash_*^Q(\varphi\to\psi)\vee\bot$ and $\forall x(\varphi\vee\bot)\vdash_*^Q(\forall x\varphi)\vee\bot$ for all $\varphi,\psi$. 
			\end{enumerate}
		\end{enumerate}
	\end{application}
	
	\subsection{The deduction nucleus}\label{deduction}
	Let $\vdash$ be $\vdash_p$ or $\vdash_p^Q$ plus additional axioms. 
	We fix a propositional formula $A$ and set 
	\[j\varphi\equiv A\to\varphi\,.\]
	This $j$, which we call the \emph{deduction nucleus}, is an instance of the \emph{open nucleus} \cite{van:kuroda,rosen}.
	As for this $j$ stability (\ref{extaxiom}) is equivalent to $\vdash A$, the strong extension $\vdash^j$ 
	is the smallest extension of $\vdash$ in which $A$ is derivable. 
	
	The deduction nucleus is compatible with substitution, and the following axioms are easy to show 
	(see, e.g., \cite[Lemma 6.2.1]{dvd:las} for the case of intuitionistic logic): 
	\begin{align*}
		\varphi\to(A\to\psi)&\vdash A\to(\varphi\to\psi)\\
		\forall x(A\to\varphi)&\vdash A\to\forall x\varphi
	\end{align*}
	Hence we get the following instance of Theorem \ref{mainlogic} where ${\vdash}={\vdash_{*}}$ is $\vdash_p$ or $\vdash_p^Q$ plus additional axioms:
	\begin{application}\label{open}
		Let $\vdash$ be $\vdash_p$ or $\vdash_p^Q$ plus additional axioms. Then
		\[\Gamma\vdash^j\varphi\iff\Gamma\vdash A\to\varphi\]
		that is, $A\to\varphi$ is derivable from $\Gamma$ if and only if $\varphi$ is derivable from $\Gamma$ when assuming that $A$ is derivable. 
	\end{application}
	As $\Gamma\vdash^j\varphi$ also means that $\varphi$ is derivable from $\Gamma\cup\{A\}$, Application \ref{open} is a variant of the \emph{deduction theorem}: 
	\[\Gamma,A\vdash\varphi\iff\Gamma\vdash A\to\varphi\]
			
	\section{Related and future work}
	
	In propositional lax logic (PLL) \cite{Fair97} the modality $\bigcirc$ is characterised by axioms and rules corresponding \cite[p.~2, (2)]{Fair97} 
	to the ones of a (logical) nucleus. Also the rules L$j$ and R$j$ of the present paper are counterparts of the rules $\bigcirc$L and $\bigcirc$R 
	of PLL \cite[p.~5]{Fair97}. We expect to gain insight by relating our approach to PLL, its semantics and applications. To start with, 
	in the vein of \cite[Lemma 2.1]{Fair97} rule R$j$ is tantamount to the inverse of L$j$. 
	
	It is known that in certain cases, given a nucleus $j$, it is possible to define a function $J\colon S\to S$, known as \emph{Kuroda-style $j$-translation}, 
	such that $U\rhd^jb$ implies $JU\rhd_jJb$, which can be viewed as conservation of $\rhd^j$ over $\rhd_j$ modulo $j$-translation. 
	Some particular instances of this are discussed in \cite{van:kuroda}, Proposition 4. Is there a general result for arbitrary entailment relations?
	
	It will be a challenge to include also other proof translation methods. For instance, 
	Friedman's A-translation \cite{fri:translation} makes use of the closed nucleus to prove Markov's rule; 
	and Ishihara and Nemoto \cite{DBLP:journals/mlq/IshiharaN16} use the same translation but work with 
	the open nucleus to prove the independence-of-premiss rule.
	
	We will further study nuclei about other forms of negation: weak negation over positive logic \cite{blkov2019proof},
	co-negation over dual logics \cite{Bellin2014} and strong negation over extensions of intuitionistic logic \cite{Kracht1998,doi:10.3166/jancl.18.341-364}.
	
	\section*{Acknowledgements} 
	{The present study was carried out within the projects ``A New Dawn of Intuitionism:
			Mathematical and Philosophical Advances'' (ID 60842) funded by the John Templeton Foundation,
			and ``Reducing complexity in algebra, logic, combinatorics - REDCOM'' belonging to the programme
			``Ricerca Scientifica di Eccellenza 2018'' of the Fondazione Cariverona. 
			Both authors are members of the ``Gruppo Nazionale per le Strutture Algebriche, Geometriche e le loro Applicazioni''  (GNSAGA)
			of the Istituto Nazionale di Alta Matematica (INdAM).
			Last but not least, the authors are grateful to Daniel Wessel for his ideas, interest and suggestions, 
			and to Tarmo Uustalo for pointing out propositional lax logic. 
		}
	
	\bibliographystyle{eptcs}
	\bibliography{synergy}
\end{document}